\documentclass[pra,twocolumn,superscriptaddress,aps]{revtex4-1}
\PassOptionsToPackage{usenames,dvipsnames}{color}
\usepackage{amsmath,amsfonts,amssymb,amsthm,graphicx,bbm,enumerate,times,color}
\usepackage{mathtools}
\usepackage{import}
\usepackage{hhline}
\usepackage{hyperref}

 \definecolor{jens}{rgb}{.2,0.7,.9}

\definecolor{darkpastelgreen}{rgb}{0.01, 0.75, 0.24}

 \newtheorem{theorem}{Theorem}
 
 \newtheorem{conjecture}{Conjecture}

 \newtheorem{lemma}[theorem]{Lemma}
 \newtheorem{corollary}[theorem]{Corollary}

\definecolor{henrik}{rgb}{1,.4,0}
\definecolor{purple}{rgb}{1,0,1}

\newcommand{\blue}[1]{{\color{blue}#1}}

\newcommand{\red}[1]{{\color{red}#1}}
\newcommand{\mc}[1]{\mathcal{#1}}

\newcommand{\mb}[1]{\mathbb{#1}}

\newcommand{\tr}{\mathrm{Tr}} 

\newcommand{\id}{\mb{1}}

\renewcommand{\mod}{\mathrm{\, mod\, }}

\newcommand{\one}{\mathbf{1}}
\newcommand{\1}{\mathrm{id}}

\newcommand{\N}{\mb{N}}

\newcommand{\R}{\mb{R}}

\renewcommand{\1}{\id}

\newcommand{\norm}[1]{\left\Vert #1 \right\Vert}

\newcommand{\ket}[1]{\left.\left|{#1}\right.\right\rangle}

\newcommand{\bra}[1]{\left.\left\langle{#1}\right.\right|}

\newcommand{\ketbra}[2]{\ket{#1} \!\! \bra{#2}}

  \newcommand{\proj}[1]{\ketbra{#1}{#1}}

\newcommand{\fu}{Dahlem Center for Complex Quantum Systems, Freie Universit{\"a}t Berlin, 14195 Berlin, Germany}
\newcommand{\ethz}{Institute for Theoretical Physics, ETH Zurich, 8093 Zurich, Switzerland}
\newcommand{\iqoqi}{Institute for Quantum Optics and Quantum Information,
Austrian Academy of Sciences, Boltzmanngasse 3, A-1090 Vienna, Austria}
\newcommand{\perimeter}{Perimeter Institute for Theoretical Physics, Waterloo, ON N2L 2Y5, Canada}
\begin{document}
\title{Von~Neumann entropy from unitarity}
 
\author{Paul Boes}
\affiliation{\fu}
\author{Jens Eisert}
\affiliation{\fu}
\author{Rodrigo Gallego}
\affiliation{\fu}
\author{Markus P. M\"uller}
\affiliation{\iqoqi}
\affiliation{\perimeter}
\author{Henrik Wilming}
\affiliation{\fu}
\affiliation{\ethz}
\begin{abstract}	
	The von~Neumann entropy is a key quantity in quantum information theory and, roughly speaking, quantifies the 
	amount of quantum information contained in a state when many identical and independent (\emph{i.i.d.}) copies of the state are available,
	in a regime that is often referred to as being asymptotic.
	In this work, we provide a new operational characterization of the von~Neumann entropy which neither requires 
	an \emph{i.i.d.} limit nor any explicit randomness. 
	We do so by showing that the von~Neumann entropy fully characterizes single-shot state transitions in 
	unitary quantum mechanics, as long as one has access to a catalyst --- an ancillary system that can be re-used after the transition --- and an environment which has the effect of dephasing in a preferred basis.
	Building upon these insights, we formulate and provide evidence for the \emph{catalytic entropy conjecture}, 
	which 
	states that the above result holds true even in the absence of decoherence. If true, this 
	would prove an intimate connection between single-shot state transitions in unitary quantum mechanics and the von Neumann entropy. 
	Our results add significant support to recent insights that, contrary to common wisdom, the standard von~Neumann entropy also characterizes single-shot situations and opens up the possibility for operational single-shot interpretations of other standard entropic quantities.
	We discuss implications of these insights to readings of the third law of quantum thermodynamics and 
	hint at potentially profound implications to holography.
\end{abstract}
\maketitle
	In quantum information theory it is common to distinguish tasks as falling in one of two regimes: Either one deals with situations in which many identically and independently distributed (\emph{i.i.d.}) quantum systems 
%
%
appear. This regime is usually referred to as the \emph{asymptotic regime}. Such tasks include, for example, Schumacher compression \cite{Schumacher1995}, entanglement distillation \cite{Bennett1996} and quantum hypothesis testing \cite{Hiai1991,Ogawa2005}. Or, in sharp 
contrast, one deals with situations that only involve a single quantum system, the so-called \emph{single-shot} regime. Examples of protocols that have been analyzed in the single-shot setting include the decoupling of quantum systems~\cite{Majenz2017}, hypothesis testing \cite{Mosonyi2015}, and state transitions in quantum thermodynamics~\cite{Brandao2013}.
Common wisdom has it that different quantities characterize these two regimes. In the first regime, the von~Neumann entropy (vNE) or quantities directly related to it prevail, such as the standard quantum relative entropy or mutual information, while in the second regime quantities such as quantum R\'{e}nyi divergences~\cite{Datta2009,Berta2015,Mueller-Lennert13,Wilde13} and smoothed versions of the above \cite{Renner2005,Datta} become important.

This common wisdom is, however, recently being challenged~\cite{Mueller2016,Lostaglio2015b,Gallego2015,Wilming2017a,Mueller2017,EntanglementFluctuationTheorem}, 
as it has been shown that the vNE determines possible single-shot state transitions in quantum mechanics --- under unitary evolutions --- provided that three assumptions hold~\cite{Mueller2017}: i) one can prepare a suitable \emph{catalyst}, i.e.\ an auxiliary system that does not change its state during the process but might become correlated with the system on which the transition is performed; ii) one has access to an environment, or source of randomness, that is modelled as a large system in the maximally mixed state; iii) one has full control over system,  catalyst and the environment, in the sense that one can implement any unitary on the joint system. Now, while introducing a catalyst as described by i) is operationally justified since it can subsequently be reused to perform transitions on further new systems, assumption ii) assigns an undesirably special role to maximally mixed systems, while assumption iii)  is in conflict with the common experience that environments cannot practically be accessed with full degree of control.

In this work, we provide an operational characterization of the von~Neumann entropy in terms of single-shot state transitions that does without assumptions ii) and iii). This may be seen as remarkable that a characterization is possible without resorting to ii) and iii) whatsoever.
Instead, our characterization builds upon two natural classes of dynamics in quantum mechanics: controlled unitary evolution and uncontrolled decoherence to some given preferred basis. We also present applications of this characterization related to  notions of
cooling in quantum thermodynamics in a way as is usually discussed in the context of quantum readings of the \emph{third law of
thermodynamics} and discuss possible implications of our results for recent work on the decoupling of systems and the AdS/CFT correspondence
in the context of \emph{holography}. Finally, we formulate, and provide evidence for, a conjecture, which, if true, shows that the von~Neumann entropy \emph{can be derived directly from unitary quantum mechanics alone} as it fully characterizes catalytic, single-shot state transitions.  

\emph{Main result.} We will now present our main result and then discuss its implications. To state the result, let $\mc D$ be the quantum channel that decoheres a system in a given orthonormal basis $\{\ket{j}\}$ of its Hilbert space, according to
\begin{align}
	\mc D[\sigma] = \sum_j \langle j|\sigma|j\rangle \proj{j}.
	\nonumber
\end{align}
Density matrices diagonal in $\{|j\rangle\}$ will be called \emph{quasi-classical}. 
Our main result can be stated as follows. 
\begin{theorem}[Single-shot characterization of the von Neumann entropy]
\label{the:main}
Let $\rho$ and $\rho'$ be two density matrices of the same finite dimension and with different spectra. Then the following two statements are equivalent:
\begin{enumerate}[i)]\label{thm:mainresult}
	\item\label{cond:mainfirst} $S(\rho')>S(\rho)$ and $\mathrm{rank}(\rho')\geq \mathrm{rank}(\rho)$.
	\item\label{cond:mainsecond} There exists a finite-dimensional, quasi-classical density matrix $\sigma$ and a unitary $U$ such that
		\begin{align}
			\label{eq:transition}\tr_2\left[U(\rho\otimes \sigma ) U^\dagger\right] &= \rho',\\ 
			\label{eq:catalytic} \mc D\left[\tr_1\left[U(\rho\otimes \sigma) U^\dagger\right]\right] &= \sigma.
		\end{align}
		\end{enumerate}
\end{theorem}
The proof is presented in Appendix~\ref{app:proofs}. Note first that the choice of basis $\{|j\rangle\}$ is irrelevant, since any basis change can be included in $U$. Furthermore, if one has $S(\rho')>S(\rho)$ but $\mathrm{rank}(\rho')<\mathrm{rank}(\rho)$, then by Theorem~\ref{thm:mainresult} the transition is not possible exactly. However, it can be done to arbitrary precision, since any state can be arbitrarily well approximated by a state with full rank. From a physical point of view, the condition on the rank is therefore not important.

To interpret this result, one can imagine a situation in which only a small region of space, say, the laboratory, can be controlled unitarily with high degree of precision while any system outside this region is decohered very quickly in some given basis. This is a common situation in current experimental devices. Given these constraints, the goal is to transform a quantum system from $\rho$ to $\rho'$ by acting unitarily on this system together with an ancillary system in a quasi-classical state that one can ``borrow'' from the environment so long as, upon being returned to the environment, it decoheres back to its initial state and can hence be used to aid further transitions. Then, Theorem~\ref{thm:mainresult} says that the vNE fully characterizes possible transitions in this natural setup (see Fig.~\ref{fig:comparison} for a comparison of results and settings).
\begin{figure}
\includegraphics[width=0.48\textwidth]{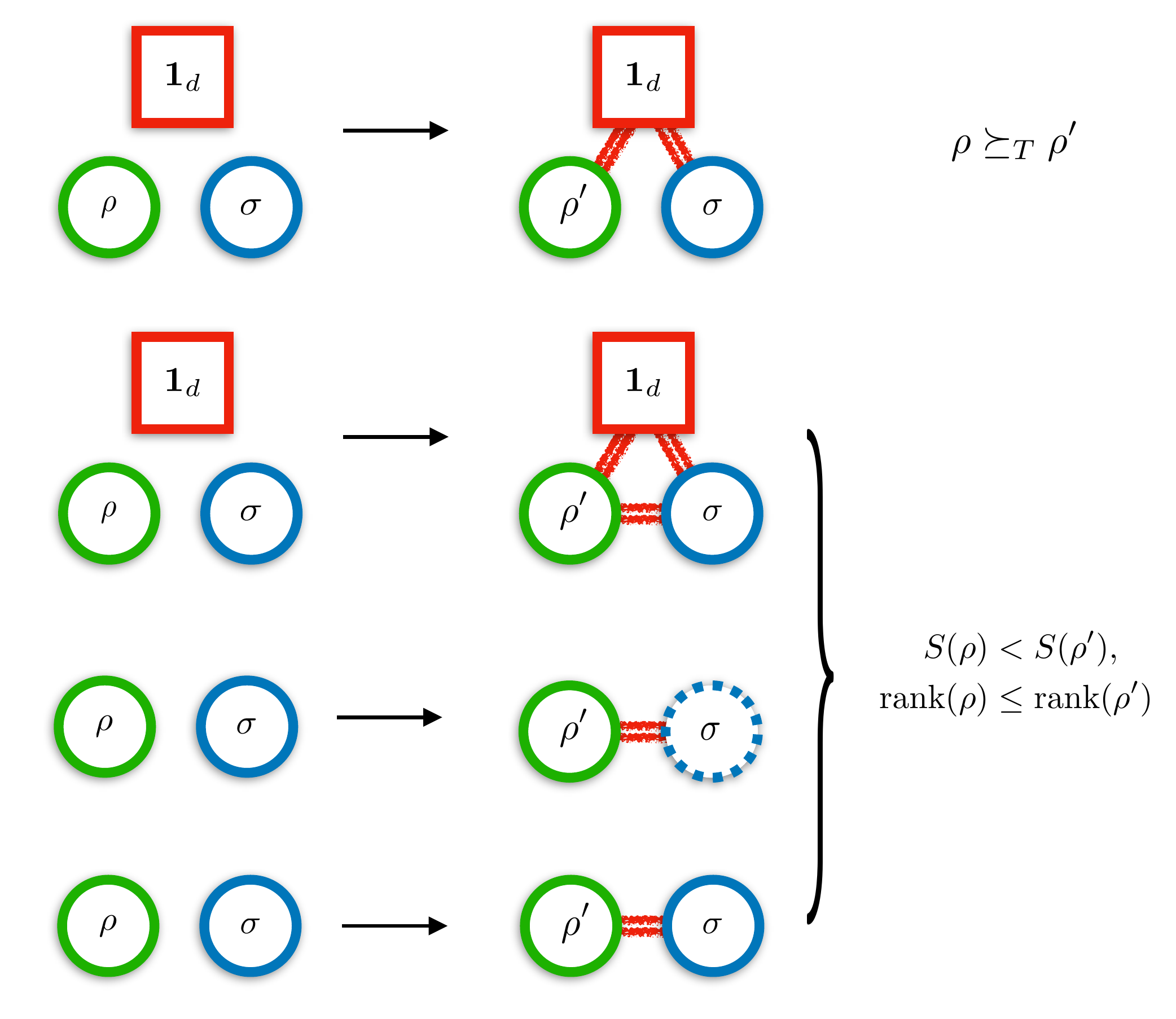}
\caption{Comparison of various settings and results. \emph{Top:} State-transitions implementable using a source of randomness and an uncorrelated catalyst $\sigma$ are characterized by the trumping relations. \emph{Middle Top:} State transitions allowing for source of randomness and a correlated catalyst are characterized by entropy and rank \cite{Mueller2017}. \emph{Middle Bottom:} By Theorem~\ref{thm:mainresult}, state transitions using a correlated catalyst and a dephasing environment that acts on the catalyst (dashed boundary) are also characterized by entropy and rank. \emph{Bottom:} State transitions using a correlated catalyst alone are characterized by entropy and rank. This is the content of Conjecture~\ref{conj:catalytic}.}
\label{fig:comparison}
\end{figure}

Note finally that, in general, the auxiliary system is clearly necessary to implement the transition $\rho\rightarrow\rho'$ since otherwise we would act unitarily on $\rho'$ and therefore could not change its spectrum. The same restriction would arise if we demanded that the auxiliary system is returned \emph{uncorrelated} from the system. Thus, $\sigma$ truly acts like a catalyst by enabling transitions that would otherwise be impossible and can, after decohering in the environment, catalyse further transitions $\rho\rightarrow \rho'$ on further independent copies of $\rho$. 
At the same time, the statement provides a new perspective to the crucial role of correlations between the system and the catalyst.

\emph{Applications to notions of cooling and the third law.} We now discuss an application of Theorem~\ref{thm:mainresult} to one of the 
key problems in quantum thermodynamics. Namely, we analyze how it can be used as a protocol for \emph{cooling to very low temperatures} 
beyond the \emph{i.i.d.} setting. This is a situation usually
captured in readings of the \emph{third law of thermodynamics} 
or \emph{Nernst's Unattainability Principle (UP)}, bounding achievable rates to cooling. Specifically,
in this context, we consider the reading of the problem of preparing systems in a state which is arbitrarily close to being pure. 
Let us for simplicity take as an initial system two uncorrelated qubits $\rho=\varrho \otimes \varrho$ with $S(\varrho)<1/2$
(even the generalization to other systems is obvious). 
Theorem~\ref{thm:mainresult} then implies that it is possible to implement a transition satisfying~\eqref{eq:transition} and \eqref{eq:catalytic} so that the final state is $\rho'= \varrho' \otimes \one_2$, where $\one_k$ represents a maximally mixed state of dimension $k$ and $\varrho$ is any full-rank state with $S(\varrho') = \epsilon$ for arbitrarily small $\epsilon > 0$, i.e.\ arbitrarily close, in trace distance, to a pure state. 
This is reminiscent of protocols of \emph{algorithmic cooling} \cite{Schulman1999,Schulman2005,Boykin2002,Raeisi2015} which take a large number $n$ of ``warm'' qubits $\varrho$ and distill from them $n_c=n(1-S(\varrho))$ ``cold'' qubits having each a smallest eigenvalue $\lambda_{\text{min}}=\mathcal{O}(\exp(-n))$ (see in particular Ref. 
\cite{Schulman1999}). 
The advantage of our protocol employing a catalyst is that we can obtain \emph{arbitrarily cold systems} using a small number of copies, $n=2$ in this case, in contrast to the asymptotic \emph{i.i.d.} setting considered in algorithmic cooling. 
Furthermore, the fact that the protocol of Theorem \ref{thm:mainresult} is catalytic allows one to repeat the protocol for $n/2$ copies of $\rho$ using a single ancillary system. 
Taking $S(\varrho) \approx 1/2$ we obtain $n_c\approx n/2$ qubits which are arbitrarily close to a pure state. 
This coincides with the bound given by algorithmic cooling which in this case is $n_c=n(1-S(\varrho)) \approx n/2$ and that is the ultimate bound for any entropy non-decreasing protocol. 
Hence, our protocol not only distills arbitrarily cold qubits with few copies, but also has an optimal efficiency ---in terms of the rate of almost pure qubits--- when applied sequentially in the asymptotic limit. 
At the same time, however,  our protocol establishes correlations among the cold qubits produced. 
Hence, although they can be used individually for further applications, it would be wrong to conclude that using our results one can prepare an arbitrary number $(\varrho')^{\otimes n}$ of uncorrelated quasi-pure states using the same catalyst over and over (see Appendix~\ref{app:cooling} for further discussion of this point). This again stresses the importance of correlations in the scheme.

The fact that one can produce systems in a state $\varrho'$ which is arbitrarily close to a pure state might, moreover, at first glance seem to be in contradiction with the third law of thermodynamics as formulated in the UP. The UP states that infinite time is required to cool down a system to its ground state (see, e.g., Refs. \cite{Levy2012,Scharlau2016,Masanes2017,Wilming2017} for recent approaches to quantum readings 
of the UP and their relation with pure state preparation). However, we note that preparing an arbitrarily pure $\varrho'$ requires also an arbitrarily large catalyst $\sigma$ and might also require a very large environment to implement the dephasing map $\mathcal{D}$, which in turn ensures that it cannot be prepared in finite time.
%
%
%

\emph{Relation to previous work.} Let us now briefly discuss the relation of our results to previous work. To begin with, we note that one can use previous results to fully characterize the possible state transitions $\rho \to \rho'$ for the special case in which the catalyst is constrained to be a maximally mixed state. Specifically, one can recast recent results \cite{Gour2015,Boes2018a} as the statement that there exist $d$ and $U$ such that
\begin{align}
 \tr_2 [U (\rho \otimes \one_d) U^{\dagger}]&=\rho' ,\\
\mathcal{D} [\tr_1 [U (\rho \otimes \one_d ) U^{\dagger}]]&=\one_d,
\end{align}
if and only if $\rho$ \emph{majorizes} $\rho'$, denoted by $\rho \succeq \rho'$ \cite{Gour2015}. Clearly, the above is a special case of Eqs.~\eqref{eq:transition} and \eqref{eq:catalytic}. Majorization captures the state transitions that are possible under random unitary evolution and hence the above establishes the intuitive result that every random unitary evolution can be implemented with a sufficiently large source of randomness without affecting the latter's state.  

To compare this result with Theorem~\ref{thm:mainresult} it should be noted that $\rho \succeq \rho'$ is, as a constraint, much stronger than $S(\rho') >S(\rho)$. 
Indeed one can see that R\'{e}nyi entropies $S_{\alpha}$, defined as
\begin{align}
S_\alpha(\rho) = \frac{1}{1-\alpha} \log \tr(\rho^\alpha) \quad (\alpha \in \mathbb{R}\backslash \{1\}),
\end{align}
cannot decrease for transitions $\rho \to \rho'$ with $\rho \succeq \rho'$, where the vNE is given by the particular case of $S \equiv S_1 := \lim_{\alpha \to 1} S_\alpha$.
The infinite set of conditions given by the R\'{e}nyi entropies
\begin{align}\label{eq:trumping}
S_{\alpha}(\rho') \geq S_{\alpha}(\rho) \: \: \forall \: \alpha \in \mathbb{R}
\end{align}
become both necessary and sufficient for the existence of a further auxiliary system $\sigma$ such that $\rho \otimes \sigma \succeq \rho' \otimes \sigma$ --- an important relation known as \emph{trumping}~\cite{Klimesh2007,Turgut2007} in quantum information theory. The trumping constraints lie, in strength, strictly between those imposed by majorization and the vnE alone.

Lastly, in Ref.\ \cite{Mueller2017} it is shown that by allowing for correlations between both systems it is possible to collapse the infinite set of conditions for the trumping conditions to essentially the vNE. In particular, it is shown that condition $i$) in Theorem~\ref{thm:mainresult} is equivalent to the existence of $\sigma$ and $U$ so that $\rho \otimes \sigma \succeq \rho'\sigma$, where $\rho'\sigma$ denotes a density matrix such that $\tr_2(\rho'\sigma)=\rho'$ and $\tr_1(\rho'\sigma)=\sigma$. This statement differs from Theorem~\ref{thm:mainresult} in that one needs to make use of a maximally mixed system over which one has full unitary control, while Theorem~\ref{thm:mainresult} includes external randomness only in the form of an uncontrolled dephasing map (see Fig.~\ref{fig:comparison} for comparison). 

\emph{Catalytic entropy conjecture.} The discussion above raises the natural question whether an external environment, being modelled as a maximally mixed state or a dephasing map as above, is at all necessary to implement all transitions which do not decrease the vNE. This is what we capture in the following conjecture.

\begin{conjecture}[Catalytic entropy conjecture]\label{conj:catalytic} 
Let $\rho$ and $\rho'$ be two density matrices of the same finite dimension and with different spectra. Then the following two statements are equivalent:
\begin{enumerate}[(a)]
	\item $S(\rho')>S(\rho)$ and $\mathrm{rank}(\rho')\geq \mathrm{rank}(\rho)$.\label{cond:conjecture_1}
	\item \label{cond:conjecture_2} There exists a density matrix $\sigma$ and a unitary $U$ such that
		\begin{align}\label{eq:conjecture_transition}
			\tr_2\left[U(\rho\otimes \sigma ) U^\dagger \right] = \rho'\ \text{and}\  \tr_1\left[U(\rho\otimes \sigma )  U^\dagger\right] = \sigma. 
		\end{align}
\end{enumerate}
\end{conjecture}
The implication \eqref{cond:conjecture_2} $\Rightarrow$ \eqref{cond:conjecture_1} follows directly from the sub-additivity of the vNE and $S_0$, hence the real content of the conjecture is that \eqref{cond:conjecture_1} are the only constraints on transitions of the form \eqref{cond:conjecture_2}. If true, this conjecture would therefore imply that the von~Neumann entropy characterizes catalytic state-transitions in unitary quantum mechanics in full generality, without the need to introduce noise or \emph{i.i.d.} limits (see Fig.~\ref{fig:comparison}). 

Let us now discuss why we believe this conjecture to be true. To begin with, it is easy to generate counterexamples that rule out the possibility that transitions of the form \eqref{cond:conjecture_2} are constrained by the aforementioned trumping relations. In Fig.~\ref{fig:example} we provide such a counterexample together with a method to construct further examples. But in fact, we can rule out more general constraints than~\eqref{eq:trumping} with the help of the following lemma.
\begin{figure}[t]
\begin{tabular}{ c  c | c}
	$[\gamma_{A,B}]_{0,0}$ & $[\gamma_{A,B}]_{0,1}$ & $[\gamma_{A}]_0$\\
	$[\gamma_{A,B}]_{1,0}$ & $[\gamma_{A,B}]_{1,1}$ & $[\gamma_{A}]_1$\\
	$[\gamma_{A,B}]_{2,0}$ & $[\gamma_{A,B}]_{2,1}$ & $[\gamma_{A}]_2$\\
	\hline
	$[\gamma_{B}]_0$ &$[\gamma_{B}]_1$&   
\end{tabular}
\quad ; \quad
\begin{tabular}{ c  c | c}
	\red{0} & 0 &0\\
	\red{$\frac{2}{6}$} & \blue{$\frac{1}{6}$} &$\frac{1}{2}$\\
	$\frac{2}{6}$ & \blue{$\frac{1}{6}$}  & $\frac{1}{2}$\\
	\hline
	$\frac{2}{3}$ &$\frac{1}{3}$&   
\end{tabular}
\quad$\rightarrow$ \quad
\begin{tabular}{ c  c | c}
	\blue{$\frac{1}{6}$}& 0 &$\frac{1}{6}$\\
	\blue{$\frac{1}{6}$} & \red{0} &$\frac{1}{6}$\\
	$\frac{2}{6}$ &  \red{$\frac{2}{6}$} & $\frac{2}{3}$\\
	\hline
	$\frac{2}{3}$ &$\frac{1}{3}$&   
\end{tabular}
\caption{Given an arbitrary bipartite state on $A,B$ denoted $\gamma_{A,B}$, the table at the left-hand side indicates the meaning of each entry, where $[\gamma_{A,B}]_{i,j}\coloneqq \bra{i,j}\gamma_{A,B} \ket{i,j}$ on a given computational basis of $AB$. The two tables at the right hand side indicate a particular transition of the form $\rho \otimes \sigma \to U (\rho \otimes \sigma)U\coloneqq \rho'\sigma$. In this case we take $\rho$ and $\sigma$ to be of dimension 3 and 2 respectively, and both diagonal in the computational basis. The unitary $U$ is simply a classical permutation which swaps the red entries with the blue entries. Note that the final state satisfies $\tr_{2}(\rho'\sigma)=\sigma$ since the bottom row remains unchanged, as demanded by condition \eqref{cond:conjecture_2}. The column sums on the right-hand side of each table represent $\rho={\rm diag}(0,1/2,1/2)$ and $\rho'={\rm diag}(1/6,1/6,2/3)$. Since $S_\infty(\rho)$ is determined by the largest eigenvalue of $\rho$, this example realizes a catalytic transition $\rho\to\rho'$ with $S_\infty(\rho) > S_\infty(\rho')$ and hence excludes the possibility that catalytic state transitions are constrained by the trumping relations.} 
\label{fig:example}
\end{figure}

\begin{lemma}[Weak solution to catalytic entropy conjecture]\label{lemma:weaksolution} 
Let $\rho$ and $\rho'$ be two density matrices of the same, finite dimension and with different spectra. Then the following two statements are equivalent:
\begin{enumerate}[(I)]
	\item \label{cond:weak_1} $S(\rho')>S(\rho)$ and $\mathrm{rank}(\rho')\geq \mathrm{rank}(\rho)$.
	\item \label{cond:weak_2} There exists a density matrix $\sigma$, a unitary $U$ and some finite dimension $d$ such that
		\begin{align}
			&\tr_2\left[U(\rho\otimes \one_d\otimes \sigma)  U^\dagger\right] = \rho'\otimes \one_d,\\  &\tr_1\left[U(\rho\otimes \one_d
			\otimes\sigma ) U^\dagger\right] = \sigma. 
		\end{align}
\end{enumerate}
\end{lemma}
This result, which is proven in Appendix~\ref{app:proofs}, supports the conjecture in two ways: Firstly, it shows that the catalytic entropy conjecture is true up to an additional maximally mixed system that remains uncorrelated to the system of interest, but not to the catalyst. 
It can also be seen as an instance of the full catalytic entropy conjecture for the specific states $\rho\otimes \one_d$ and $\rho'\otimes \one_d$.
Secondly, and more importantly, it allows us to prove the following corollary:

\begin{corollary}[Characterization of entropy functions] \label{cor:exclude}
Let $f$ be a function from the set of density matrices to the real numbers such that for every transition of the form \eqref{cond:conjecture_2} between full-rank density matrices, $f(\rho') > f(\rho)$. Then exactly one of the following two statements is true:
\begin{enumerate}
	\item $S(\rho') > S(\rho) \Leftrightarrow f(\rho') > f(\rho)$,
    \item $f$ is non-additive or discontinuous.\label{cond:nonadditive}
\end{enumerate} 
\end{corollary}
Corollary~\ref{cor:exclude} follows from Lemma~\ref{lemma:weaksolution} by showing that any such function $f$ has to be a linear function of the vNE (see Appendix~\ref{appSunique} for a proof). Thus, for full-rank density matrices, if Conjecture~\ref{conj:catalytic} was false, any additional constraint on transitions of the form~\eqref{cond:conjecture_2} would have to be given by exotic entropic functions that are not additive or are discontinuous. For instance, this corollary immediately implies that none of the functions $S_\alpha, \:\alpha\neq 0,1$, can be a monotone for transitions of the form \eqref{cond:conjecture_2} since they all satisfy none of the two conditions in the corollary.

\emph{Discussion and open questions.} 
	In this work, we have provided a new operational characterization of von~Neumann entropy which adds significant support to recent proposals that, contrary to common wisdom, the standard von Neumann entropy characterizes not only the \emph{i.i.d.} limit but also single-shot protocols in quantum information theory. We have done so by showing that the von~Neumann entropy fully determines the possibility of single-shot state transitions in unitary quantum mechanics, as long as one has access to a catalyst and environmental dephasing in a preferred basis. Furthermore, we have formulated the \emph{catalytic entropy conjecture} which essentially states that the above result holds true even in the absence of decoherence. We have also presented evidence for the truth of this conjecture by ruling out alternatives.

Our work suggests that there might be a novel, hitherto unexplored sector of quantum information theory in which operations on \emph{single} copies of a quantum state are characterized directly in terms of standard entropic quantities like vNE. For example, one may ask what happens in Theorem~\ref{the:main} or Conjecture~\ref{conj:catalytic} if we introduce another reference system $R$ that is initially correlated or entangled with the system $1$ (let us denote system $1$ by $A$ for now, and let $C$ be the catalytic system $2$). Applying a unitary $U_{A,C}$ on the system and catalyst, denoting the new states of the systems by $R'$, $A'$ and $C'$, we obtain $R'=R$, by construction $C'=C$ and $S(A')\geq S(A)$ since $A$ becomes correlated with $C$. Furthermore, the mutual information $I(R:A)=S(R)+S(A)-S(R, A)$ satisfies $I(R':A')\leq I(R:A)$. Are these necessary conditions also \emph{sufficient} for the existence of a transformation of that form --- in particular, can $A$ retain almost all of its correlations with $R$ under correlating-catalytic transformations? A positive answer to this or other similar questions would yield a new single-shot interpretation of the standard mutual information which could potentially be useful in the context of \emph{decoupling}~\cite{Horodecki2005,HaydenTutorial2011,Dupuis2014,Majenz2017} or merging of quantum states.
    
 The results also hint at the insight that entanglement in single many-body systems can well be captured in terms of the von-Neumann entropy. Ideas on \emph{single-copy entanglement} have been considered in situations where each 
specimen consists of a many-body system, already naturally featuring asymptotically many constituents \cite{PhysRevA.72.042112}. 
Then it can be unreasonable to capture entanglement of subsystems in yet another
asymptotic limit of many copies of identical quantum many-body systems. The results laid out here give substance to 
the intuition that even in single specimens
of quantum many-body systems,  entanglement can in this context be quantified in terms of the familiar von-Neumann entanglement entropy.
    
Results of this kind would also have implications in the context of \emph{holographic approaches} 
to quantum gravity, as in the AdS/CFT correspondence (see, for example, Refs.~\cite{Susskind1995,Maldacena1999,Ryu2006,Hayden2013,Czech2015,Lashkari2015,Casini2017,Jahn}). In these approaches, standard von Neumann (entanglement) entropies of boundary regions turn out to correspond to geometric quantities of a dual gravity theory in the bulk. 
In fact, it is exactly the mutual information that we have just discussed which is believed to be directly related to geometric quantities like area also in other (non-AdS/CFT) approaches to emergent spacetime~\cite{Cao2017}. 
To shed some light on this correspondence, it is therefore natural to consider operational interpretations of entropy in the boundary theory, and to ``dualize'' them to obtain corresponding interpretations of geometric quantities in the bulk. 
A difficulty in doing so, however, is that the protocols on the boundary theory either involve many copies of the state (which seems unphysical given that there is a unique spacetime), or they lead to quantification in terms of single-shot entropies (see, e.g., Ref.~\cite{Czech2015}) which do not always have a direct dual interpretation. 
The proven and conjectured results of this paper could therefore resolve this difficulty, by supplying direct single-shot interpretation of standard entropic quantities which might ultimately shed some light on the operational basis of geometric quantities. It is the hope that the present work stimulates
such endeavors.

\emph{Acknowledgements.} We acknowledge funding from DFG (GA 2184/2-1, CRC 183, EI 519/14-1, EI 519/9-1, FOR 2724), the ERC (TAQ) and the Studienstiftung des deutschen Volkes. HW further acknowledges contributions from the Swiss National Science Foundation via the NCCR QSIT as well as project No. 200020\_165843.
This research was supported in part by Perimeter Institute for Theoretical Physics. Research at Perimeter Institute is supported by the Government of Canada through the Department of Innovation, Science and Economic Development Canada and by the Province of Ontario through the Ministry of Research, Innovation and Science.

\bibliographystyle{apsrev4-1}

%

\appendix

\section{Proof of Theorem~\ref{thm:mainresult} and Lemma~\ref{lemma:weaksolution}} \label{app:proofs}

In this section we prove Theorem~\ref{thm:mainresult} and Lemma~\ref{lemma:weaksolution}. 
The proofs of both results rely on the following recent result from Ref.~\cite{Mueller2017}. 
\begin{theorem}[Correlating-catalytic majorization~\cite{Mueller2017}]\label{thm:markus} 
	Let $\rho,\rho'$ be two density matrices on the same, finite-dimensional Hilbert space $\mc H_A$ such that $S(\rho)< S(\rho')$ and $\mathrm{rank}(\rho) \leq \mathrm{rank}(\rho')$. Then there exists a density matrix $\tau$ on a finite-dimensional Hilbert space $\mc H_B$ and a bipartite density matrix $\rho'\tau$ on $\mc H_A \otimes \mc H_B$ such that
	\begin{align*}
		\rho \otimes \tau \succeq \rho'\tau, \quad \tr_B[\rho'\tau]=\rho',\quad \tr_A[\rho'\tau] = \tau. 
	\end{align*}
\end{theorem}
Another result that will be used frequently is the Schur-Horn-Theorem. 
\begin{theorem}[Schur-Horn \cite{Horn1954}]
For a matrix $H$, let $\lambda(H)$ be the vector of its eigenvalues and $\mathrm{diag}(H)$ the vector of its diagonal entries. If $H$ is Hermitian, then the following are equivalent:
\begin{itemize}
\item $\lambda(H) \succeq \mathrm{diag}(H)$,
\item there exists a unitary matrix $U$ such that 
\[
U \hat{\lambda}(H) U^\dagger = H,
\]
\end{itemize}
where $\hat{\lambda}(H)$ is the diagonal matrix with diagonal $\lambda(H)$. 
\end{theorem}
In particular, the Schur-Horn theorem implies that, if $\rho \succeq \rho'$, then there exist unitaries $U, V$ such that 
\begin{align} \label{eq:dephasing_schur}
\rho' = V \left(\mc D_J [U \rho U^\dagger ] \right) V^\dagger.
\end{align}
Here and in the following, in contrast to the main text, we explicitly denote the choice of basis $J=\{|j\rangle\}$ in the notation for the decoherence map, $\mc D=\mc D_J$. If we choose $J$ as the eigenbasis of $\rho'$ then $V$ is the identity map.
We are now in position to prove Theorem~\ref{thm:mainresult}.

\begin{proof}[Proof of Theorem~\ref{thm:mainresult}]
   We begin with proving that \ref{cond:mainfirst}) implies \ref{cond:mainsecond}).
	Thus, assume that $S(\rho)<S(\rho')$ and $\mathrm{rank}(\rho)\leq\mathrm{rank}(\rho')$. 
	Then Theorem~\ref{thm:markus} together with \eqref{eq:dephasing_schur} implies that there exists a unitary $W_{A,B}$ and two bases $J_A$ and $J_B$ such that
\[
		(\mc D_{J_A}\otimes \mc D_{J_B} )\left[W_{A,B} (\rho\otimes \tau) W_{A,B}^\dagger\right] =\rho'\tau.
\]
	From locality of quantum mechanics and the Schur-Horn theorem we thus find that
\[
		\rho'\tilde{\tau} :=  \mc (D_{J_A}\otimes \mathbb{I} )\left[W_{A,B} (\rho\otimes \tau) W_{A,B}^\dagger\right]
\]
is a quantum state with the properties $\tr_B[\rho'\tilde{\tau}] = \rho'$ and $\tilde\tau=\tr_A[\rho'\tilde\tau] \succeq \tau$. Here, $\mathbb{I}$ denotes the identity super-operator.
   
As a second step, we show that we can realize any dephasing map on a system $A$ using an ancillary system in a maximally mixed state. To see this, let $R$ be a system of the same dimension $d$ as $A$ and let $\{U_k\}_{k=1}^{d}$ be a unitary operator basis on $A$, meaning a collection of $d$ unitaries $U_k$ such that
\begin{align}
\tr\left[U_j U_k^\dagger\right] = d \delta_{j,k}. \label{eq:uob}
\end{align}
Such a set of operators exists on every finite-dimensional Hilbert-space \cite{Schwinger1960,Werner2001}. Then, define the unitary 
\[
    V_{A,R} = \sum_{j=1}^d  \proj{j}_A \otimes (U_j)_R,
\]
where we recall that $J= \{\ket{j}\}$. Now, it is easy to check that for any $\rho=\rho_A$,
\[
\tr_R \left[ V_{A,R} (\rho \otimes \mathbf{1}_{d}) V^\dagger_{A,R} \right] = \mc D_J[\rho].
\]
In a third step, we now show that we can use this dilation of the dephasing map to construct a catalyst for Theorem~\ref{thm:mainresult}. To do so, let  
\[
		\sigma := \tau \otimes \mathbf{1}_{d}
\]
	and define the unitary 
\[
	U_{A,B,R} = (V_{A,R}\otimes \mathbf{1}_B)(W_{A,B}\otimes \mathbf{1}_R). 
\]
	From the previous discussion and the construction of the dephasing unitary $V_{A,R}$, we know that 
\[
	\tr_{R}\left[U_{A,B,R} (\rho\otimes \sigma) U_{A,B,R}^\dagger \right] = \rho'\tilde{\tau}.
\]
	Thus, what is left to be proven is that $\sigma$ is indeed a valid catalyst, i.e., does not change in the course of the process except 
	from building up coherences. We will show that it undergoes the transition
\[
		\sigma = \tau \otimes \mathbf{1}_d \rightarrow \tilde\tau \otimes\mathbf{1}_d. 
\]
To show this, first note that the dephasing dilation implemented by $V_{A,R}$ leaves the state $\one_d$ of $R$ locally unchanged. But this means that we only have to show that $R$ does not become correlated with $B$ in the dephasing step, since it follows from locality that the marginal on $R$ remains unchanged and the marginal on $B$ evolves from $\tau$ to $\tilde\tau$.
To see that $B$ and $R$ remain uncorrelated, we simply compute the action of the dephasing unitary $V_{A,R}$ on $B,R$, to get
\begin{align*}
	&\tr_A\left[U_{A,B,R} (\rho\otimes \sigma) U_{A,B,R}^\dagger \right]\\ 
	&\quad= \sum_{j,k} \tr_A\left[\proj{j}_A W_{A,B} (\rho\otimes \tau) W_{A,B}^\dagger \proj{k}_A\right] \otimes \frac{U_j U_k^\dagger}{d_R} \\
	&\quad= \sum_j \bra{j}_A W_{A,B} (\rho\otimes \tau) W_{A,B}^\dagger \ket{j}_A \otimes \mathbf{1}_d\\
	&\quad = \tilde \tau \otimes\mathbf{1}_d,
\end{align*}
where we have dropped identities for notational convenience. This proves that \ref{cond:mainfirst}) implies \ref{cond:mainsecond}).

Let us now prove that \ref{cond:mainsecond}) implies \ref{cond:mainfirst}). 
In the following let $\alpha\in\{0,1\}$. 
Since $S_0(\rho) = \log(\mathrm{rank}(\rho))$, both $S_0$ and $S_1=S$ are subadditive and additive. Since the final state on the catalyst, which we now call $\sigma'$, satisfies $\mc D_J[\sigma']=\sigma$, it follows that $\sigma'\succeq \sigma$ and thus $S_\alpha(\sigma') \leq S_\alpha(\sigma)$. Furthermore, from additivity and subadditivity we get
\begin{eqnarray*}
	S_\alpha(\rho)+S_\alpha(\sigma) &=& S_\alpha (\rho\otimes\sigma)= S_\alpha(\rho' \sigma') \\
			      &\leq& S_\alpha(\rho') + S_\alpha(\sigma')\leq S_\alpha(\rho') + S_\alpha(\sigma).
\end{eqnarray*}
For $\alpha=0$, this proves $\mathrm{rank}(\rho)\leq \mathrm{rank}(\rho')$. For $\alpha=1$, equality, i.e.\ $S(\rho)=S(\rho')$, is only possible if $S(\rho'\sigma')=S(\rho')+S(\sigma')$, and it is well-known that this implies $\rho'\sigma'=\rho'\otimes\sigma'$. Thus $\rho\otimes\sigma\succ \rho'\otimes\sigma'\succ \rho'\otimes\sigma$, and so $\rho\succ_T \rho'$ for the trumping relation, which together with $S(\rho)=S(\rho')$ implies that $\rho$ and $\rho'$ have the same spectrum, i.e.\ are unitarily equivalent~\cite{Klimesh2007,Turgut2007}. This contradicts the assumptions of the theorem. We must thus have $S(\rho)<S(\rho')$, which completes the proof.
\end{proof}

Let us now turn to the proof of Lemma~\ref{lemma:weaksolution}, which builds on the proof of Theorem~\ref{thm:mainresult}. 
\begin{proof}[Proof of Lemma~\ref{lemma:weaksolution}]
	For this proof we re-use all the notation from the proof of the implication \ref{cond:mainfirst})$\Rightarrow$\ref{cond:mainsecond}) of Theorem~\ref{thm:mainresult}. 
	In particular note that the final state $\tilde\tau$ on the $B$-subsystem of the catalyst only needs to be dephased in a basis $J_B$ to be returned exactly, since, by construction, $\mathrm{diag}(\tilde\tau) = \lambda(\tau)$. Using the dephasing construction already used in the proof of Theorem~\ref{thm:mainresult} we can include a further system $R_2$ in the maximally mixed state into the system and use the dephasing ttunitary $V_{R_2B}$ at the end of the process to dephase system $B$. The only property that we still need to prove is that this does not introduce correlations between $A$ and $R_2$. However, this is exactly the same calculation that shows that there are no correlations between $B$ and $R$ at the end in the proof of Theorem~\ref{thm:mainresult}. We only have to exchange $R$ for $R_2$ and $B$ for $A$. This finishes the proof.  
	\end{proof}
	
\section{Catalytic cooling}\label{app:cooling}

Let us first present in detail how to prepare almost pure states with a protocol that uses Theorem~\ref{thm:mainresult}. 
Using this theorem, we have that, given system ${Q_1}$ in state $\rho_{{Q_1}}=\varrho \otimes \varrho$ with $2S(\varrho)< 1$, one can find $U$ and a catalyst ${C}$ in state $\sigma$ so that 
\begin{align}
\gamma_{Q_1C} =(\mathcal{D}_J \circ \mathcal{U}_1) [ \rho_{{Q_1}} \otimes \sigma_{C}]
\end{align}
where $\mathcal{D}_J$ is the map locally dephasing the system $C$ and leaving ${Q_1}$ untouched (formally $\mathbb{I}_{Q_1} \otimes \mathcal{D}_J$), and $\mathcal{U}_1[\bullet] = U \bullet U^{\dagger}$. Also, we denote by $\gamma_{Q_1C}$ a bipartite state on ${Q_1C}$ which, according to Theorem~\ref{thm:mainresult}, fulfills $\tr_{Q_1}(\gamma_{Q_1C})=\sigma_{C}$ and $\tr_{C}(\gamma_{Q_1C})=\rho'_{Q_1}=\varrho' \otimes \one_2$, where $\varrho'$ can be any full-rank state, but in the following we are interested in the case where $\rho'$ is arbitrarily close to a pure state.

This protocol can be iterated on an arbitrary number $n$ of subsystems ${Q_1,\ldots,Q_n}$, taking initially $\rho_{Q_1,\ldots,Q_N} = \rho_{Q_1} \otimes \cdots \otimes \rho_{Q_n}$ as input, where $\rho_{Q_i}=\varrho \otimes \varrho$ for all $i$. We define the unitary channels $\mathcal{U}_i$ which apply the unitary $U$ to systems $Q_iC$ and act trivially in the rest of the subsystems, that is, 
\begin{align}
\mathcal{U}_i[\bullet]= U_{Q_i C} \otimes \mathbb{I}_{|Q_iC} \: \bullet \: U_{Q_i C}^{\dagger} \otimes \mathbb{I}_{|Q_iC}. 
\end{align}
Then, applying these unitary channels, each followed by a dephasing map on $C$, one obtains 
\begin{align}
\label{eq:sequential_cooling}\gamma_{Q_1,\ldots , Q_n C} =\mathcal{D}_J \circ \mathcal{U}_n \circ \cdots \circ \mathcal{D}_J \circ \mathcal{U}_1 [\rho_{Q_1,\ldots,Q_n} \otimes \sigma]
\end{align}
where, due to Theorem~\ref{thm:mainresult}, we have 
\begin{align*}
&\tr_{|Q_i} (\gamma_{Q_1,\ldots , Q_n C})=\varrho' \otimes \one_2  \:\: \forall \:i,\\
&\tr_{|C} (\gamma_{Q_1,\ldots , Q_n C})=\sigma.
\end{align*}
Hence, with this protocol we have prepared $n/2$ subsystems whose marginal $\varrho'$ is arbitrarily close to a pure state. Note, however, that the resulting state of the compound $\gamma_{Q_1,\ldots,Q_n}$ displays correlations between its parts, hence, although each subsystem in state $\varrho'$ can be individually used ---for instance as a pure state input of a quantum computation--- the whole compound $\gamma_{Q_1,\ldots,Q_n}$ deviates from the state
\[
\tilde{\gamma}_{Q_1,\ldots,Q_n}\coloneqq  \rho'_{Q_1} \otimes \cdots \otimes \rho'_{Q_n}.
\]
This can be seen for instance by comparing the minimum eigenvalue $\lambda_\text{min}$ of both states in the limit of large $n$, which gives
\begin{align}
\label{eq:min_eigenvalue_1}\lim_{n \rightarrow \infty} \frac{ \lambda_{\min} ( \tilde{\gamma}_{Q_1,\ldots,Q_n})} { \lambda_{\min} (\gamma_{Q_1,\ldots,Q_n})}&\leq \lim_{n \rightarrow \infty} \frac{ \lambda_{\min} ( \tilde{\gamma}_{Q_1,\ldots,Q_n})}{ \lambda_{\min} (\gamma_{Q_1,\ldots,Q_nC})} \\
\label{eq:min_eigenvalue_2}&\leq \lim_{n \rightarrow \infty} \frac{ \lambda_{\min} ( \tilde{\gamma}_{Q_1,\ldots,Q_n})}{ \lambda_{\min} (\rho_{Q_1,\ldots, Q_n} \otimes \sigma) }\\
\label{eq:min_eigenvalue_3}&= \lim_{n \rightarrow \infty} \frac{(\frac 1 2 \lambda_{\min}(\varrho'))^n}{\lambda_{\min}(\varrho)^{2n} \lambda_{\min}(\sigma)}
\\
\label{eq:min_eigenvalue_4}&=0
\end{align}
where~\eqref{eq:min_eigenvalue_1} follows simply because tracing out one subsystem can only increase the minimum eigenvalue; \eqref{eq:min_eigenvalue_2} follows due to \eqref{eq:sequential_cooling}. To see this note that map $\mathcal{D}_J \circ \mathcal{U}_n \circ \cdots \circ \mathcal{D}_J$ can be implemented as a global unitary on $Q_1,\ldots,Q_n$ together with a source of randomness of sufficiently large dimension $d$ which is responsible of the dephasing. That is, there exists $V$ so that
\[
\tr (V \rho_{Q_1,\ldots,Q_n} \otimes \sigma \otimes \one_dV^{\dagger}) =\gamma_{Q_1,\ldots Q_n C}. 
\]
This implies in turn that $\rho_{Q_1,\ldots,Q_n} \otimes \sigma \succeq \gamma_{Q_1,\ldots Q_n C}$ (see for instance Ref. \cite{Gour2015}) and that
\[
\lambda_{\text{min}}(\rho_{Q_1,\ldots,Q_n} \otimes \sigma ) \leq \lambda_{\mathrm{min}}(\gamma_{Q_1,\ldots Q_n C}).
\]
Eq.~\eqref{eq:min_eigenvalue_3} follows from simple algebra. Lastly, \eqref{eq:min_eigenvalue_4} follows from the fact that $\lambda_{\text{min}}(\varrho')$ is arbitrarily small while $\lambda_{\text{min}}(\sigma)>0$. To see the latter we recall the result of Appendix F.1 from Ref. \cite{Wilming2017}, which shows that any transition of the form \eqref{eq:sequential_cooling} employing a catalyst $\sigma$ without full rank with spectrum $\{\sigma_i\}$, can be also be implemented with a full-rank catalyst $\tilde{\sigma}$ with spectrum $\{\sigma_i | \sigma_i>0\}$. In other words, we can assume without loss of generality that $\sigma$ is full rank.

\section{The classical case}\label{app:classical}

In the following, we denote the marginals of a probability distribution $r$ on $X\times Y$ by $r_X$ resp.\ $r_Y$, such that $r_X(x)=\sum_{y\in Y} r(x,y)$ and $r_Y(y)=\sum_{x\in X} r(x,y)$. This is the classical analogue of the partial trace.

\begin{conjecture}[Classical catalytic entropy conjecture]
	Let $p$ and $p'$ be two different probability distributions on a finite space of events $X$. Then the following two statements are equivalent:
	\begin{enumerate}[i)]
	\item $S(p)\leq S(p')$, where $S$ is the Shannon entropy. \label{cond:classic_1}  
	\item For every $\epsilon>0$, there exists a probability distribution $q$ on a finite space $Y$ and a permutation $P$ on $X\times Y$ such that
	\begin{align}
         \left[ P (p\otimes q) \right]_Y = q,\quad \norm{[P (p \otimes q) ]_X-p' }_1\leq \epsilon. \label{cond:classic_2}
	\end{align}
	\end{enumerate}
\end{conjecture}
There are two reasons for which we only conjecture approximability of $p'$ to arbitrary accuracy instead of perfect achievability. Firstly, in order to drop the rank condition from condition \ref{cond:classic_1}); secondly, to account for the case in which $p$ and $p'$ differ by irrational amounts. In this case, permutations only realize the transition $p \to p'$ approximately.

Note that, since the statement of the catalytic entropy conjecture is unitarily invariant on the input states, and permutations are special cases of unitary operations, a proof of the classical catalytic entropy conjecture would essentially also prove the quantum version. 
The converse, however, is not necessarily true: it is apriori possible that only the quantum formulation holds. 
Nevertheless, as in the quantum case, one can show that the Shannon entropy is essentially the unique additive monotone. 
This follows from the following classical version of Lemma~\ref{lemma:weaksolution}. 
It uses the notation ${\rm rank}(p)$ to denote the number of non-zero entries of a discrete probability distribution $p$.
\begin{lemma}[Weak solution to catalytic entropy conjecture (classical)]\label{lemma:classical_weaksolution} 
Let $p$ and $p'$ be two different probability distributions of the same, finite dimension and with rational entries. Then the following two statements are equivalent:
\begin{enumerate}[(I)]
	\item \label{cond:classical_weak_1} $S(p')>S(p)$ and $\mathrm{rank}(p')\geq \mathrm{rank}(p)$.
	\item \label{cond:classical_weak_2} There exists a probability distribution $q$ on a finite sample space $Z$, a $d$-dimensional sample space $Y$, and a permutation $P$ on $X\times Y \times Z$ such that
\[
		\left[ P (p\otimes \mathbf{1}_d \otimes  q) \right]_Z = q,\quad [P (p \otimes \mathbf{1}_d \otimes q) ]_{X,Y} = p' \otimes \mathbf{1}_d. 
\]
\end{enumerate}
\end{lemma}
Here, $\mathbf{1}_d^\top = (1/d, \dots, 1/d)$ denotes the uniform distribution on $Y$.
\begin{proof}
	We only consider the non-obvious direction, i.e.\ we show that \eqref{cond:classical_weak_1}$\Rightarrow$\eqref{cond:classical_weak_2}. According to Ref.~\cite{Mueller2017}, if condition~\eqref{cond:classical_weak_1} is satisfied, then there exists a probability distribution $\tilde{q}$ on some sample space $\tilde{Z}$ such that $p \otimes \tilde{q} \succeq p'\tilde{q}$. 
	Since $p,p'$ are rational, and so are $\tilde q$ and $p'\tilde q$, the majorization relation implies that this transition can be realized exactly with a random permutation. In other words, there exists a $\tilde{d}$-dimensional ancilla $A$ in the state $\mathbf{1}_{\tilde{d}}$ and the global permutation $P = \sum_{i=1}^{\tilde{d}}\Pi_i \otimes P_i$, where $\Pi_i$ denotes the rank-one projector onto the standard basis $\{\mathbf{e}_i\}$ of $A$, that is, 
	$\Pi_i(q) = q_i \mathbf{e}_i$, such that 
\begin{align}	
	\frac{1}{\tilde{d}} \sum_{i=1}^{\tilde{d}} P_i(p \otimes \tilde{q}) = p'\tilde{q}.
\end{align}	 
Next, choose $d = \tilde{d}$ as the dimension of $Y$ and consider the permutation 
\begin{align}
    P' = \sum_{i=1}^{d} (\Pi_i)_Y \otimes (\pi^i)_A,
\end{align}
where $\pi$ is a permutation defined by $\pi \mathbf{e}_j = \mathbf{e}_{j+1 \mathrm{mod}\ d}$. Applying both of these permutations to the total system yields 
\begin{align*}
   & P' P \left[\strut (\mathbf{1}_{\tilde{d}})_Y \otimes (\mathbf{1}_{\tilde{d}})_A \otimes p\otimes  \tilde{q}\right] \\
	= & P' \left[(\mathbf{1}_{\tilde{d}})_Y \otimes \left( \sum_{i}^d (\mathbf{e}_i)_A/d \otimes P_i(p \otimes \tilde{q}) \right)\right] \\
   = & \sum_{i,j=1}^d (\mathbf{e}_j)_Y/d \otimes (\mathbf{e}_{i + j\mod d})_A/d \otimes P_i(p \otimes \tilde{q}).
\end{align*}
	From the last expression, we see that summing over $Y$ leaves $A$ uncorrelated from both $X$ and $\tilde{Z}$, since $\sum_j \Pi_{i+j}/d^2 = \mathbf{1}_d$, and summing over $A$ leaves $Y$ and $X$ uncorrelated. Hence, by identifying $Z = A \times \tilde{Z}$ and $q = (\mathbf{1}_d)_A \otimes \tilde{q}$, the statement of the lemma follows.
\end{proof}

\section{$S$ is the only continuous additive monotone}
\label{appSunique}
Here we give a proof of Corollary~\ref{cor:exclude}. This corollary follows immediately from the following lemma, which itself has Lemma~\ref{lemma:weaksolution} as its key ingredient.

\begin{lemma}[Properties of real and additive functions]
Let $f$ be a real function on the set of all finite-dimensional density matrices	 which is \emph{continuous} (on all subsets of density matrices of fixed dimension) and \emph{additive}, i.e.\ $f(\rho\otimes\sigma)=f(\rho)+f(\sigma)$. Furthermore, suppose that $f$ is a \emph{monotone} with respect to transitions of the form (b) of Conjecture~\ref{conj:catalytic}, i.e.\ satisfaction of condition (b) implies that $f(\rho)\leq f(\rho')$. Then there exist a constant $a\geq 0$ and dimension-dependent constants $b_n\in\R$, such that
\[
   f(\rho)=a\cdot S(\rho)+b_n,
\]
with $n$ the Hilbert space dimension of $\rho$, and $b_{m,n}=b_m+b_n$.
\end{lemma}
\begin{proof}
For any density matrix $\rho$ of dimension $n$, define the negentropy $I(\rho):=\log n -S(\rho)$. Let $\rho,\rho'$ be full-rank density matrices of possibly different dimensions $n,n'$ such that $I(\rho)=I(\rho')$, then
\[
   S(\rho\otimes \one_{n'})=\log n -I(\rho)+\log n'=S(\rho'\otimes \one_n).
\]
Let $\epsilon>0$, and let $\sigma_\epsilon$ be any full-rank state of size $nn'$ such that $\|\sigma_\epsilon-\rho\otimes\one_{n'}\|<\epsilon$ and $S(\sigma_\epsilon)<S(\rho\otimes\one_{n'})$, then $S(\sigma_\epsilon)<S(\rho'\otimes \one_{n})$, hence Lemma~\ref{lemma:weaksolution} implies that there is some $d\in\N$ such that $\sigma_\epsilon\otimes\one_d\to \rho'\otimes \one_n\otimes\one_d$, where ``$\to$'' denotes that a transition of the form (b) is possible. Thus 
\[
f(\sigma_\epsilon\otimes \one_d)\leq f(\rho'\otimes\one_n\otimes\one_d), 
\]
and additivity of $f$ yields $f(\sigma_\epsilon)\leq f(\rho'\otimes\one_n)$. Since $\lim_{\epsilon\to 0}\sigma_\epsilon =\rho\otimes\one_{n'}$, and since $f$ is continuous, this implies that $f(\rho\otimes\one_{n'})\leq f(\rho'\otimes\one_n)$. Reversing the roles of $\rho$ and $\rho'$ in the above argumentation gives the converse inequality, and hence $f(\rho\otimes\one_{n'})=f(\rho'\otimes\one_n)$. Define the new real function $j(\tau):=f(\one_n)-f(\tau)$, where $n$ is the dimension of the density matrix $\tau$, then $j$ is also additive, and it vanishes on the maximally mixed states. Thus $j(\rho)=j(\rho')$.

In summary, we have shown that $j$ is constant on the level sets of $I$. Thus, there is a real function $g:[0,\infty)\to\R$ such that $j(\rho)=g(I(\rho))$ for all $\rho$. Let $x,y\in [0,\infty)$ with $x<y$, and let $\rho_x,\rho_y$ be finite-dimensional full-rank density matrices of dimensions $n_x,n_y$ with $I(\rho_x)=x$ and $I(\rho_y)=y$. Then
\begin{eqnarray*}
g(x+y)&=& g(I(\rho_x)+I(\rho_y))=g(I(\rho_x\otimes\rho_y))\\
&=& j(\rho_x\otimes\rho_y)=j(\rho_x)+j(\rho_y)\\
&=& g(I(\rho_x))+g(I(\rho_y))=g(x)+g(y).
\end{eqnarray*}
Furthermore, $S(\rho_y\otimes \one_{n_x})<S(\rho_x\otimes\one_{n_y})$, hence there is some $d\in\N$ such that $\rho_y\otimes \one_{n_x}\otimes\one_d\to \rho_x\otimes \one_{n_y}\otimes \one_d$, therefore $j(\rho_y\otimes \one_{n_x}\otimes \one_d)\geq j(\rho_x\otimes \one_{n_y}\otimes\one_d)$, and additivity implies $j(\rho_y)\geq j(\rho_x)$. It follows that $g(y)\geq g(x)$.

We thus see that $g$ is both \emph{additive} and \emph{non-decreasing}, and it is well-known (and easy to verify) that this implies that $g(x)=ax$ for some $a\geq 0$, i.e.\ $j(\rho)=a I(\rho)$. Going back to the definition of $f$, this gives us
\[
   f(\rho)=a S(\rho)+b_n,
\]
with $n$ the dimension of $\rho$ and $b_n:=f(\one_n)-a\log n$. Finally, additivity of $f$ and $\one_{m,n}=\one_m\otimes\one_n$ imply $b_{m,n}=b_m+b_n$.
\end{proof}
Note that $b_{m,n}=b_m+b_n$ does not automatically entail that $b_m$ is proportional to $\log m$ (and thus to $S_0$): there are other well-known examples of functions on the integers which are additive in this sense.

\end{document}